\documentclass[sigconf]{acmart}

\AtBeginDocument{%
  \providecommand\BibTeX{{%
    \normalfont B\kern-0.5em{\scshape i\kern-0.25em b}\kern-0.8em\TeX}}}

\copyrightyear{2021}
\acmYear{2021}
\setcopyright{acmcopyright}\acmConference[ICTIR '21]{Proceedings of the 2021 ACM SIGIR International Conference on the Theory of Information Retrieval}{July 11, 2021}{Virtual Event, Canada}
\acmBooktitle{Proceedings of the 2021 ACM SIGIR International Conference on the Theory of Information Retrieval (ICTIR '21), July 11, 2021, Virtual Event, Canada}
\acmPrice{15.00}
\acmDOI{10.1145/3471158.3472235}
\acmISBN{978-1-4503-8611-1/21/07}

\usepackage[ruled,vlined]{algorithm2e}
\usepackage{breqn}
\usepackage{algorithmic}
\usepackage{balance}

\newtheorem{theorem}{Theorem}[section]

\newcommand\figref[1]{{Figure \ref{fig:#1}}}
\newcommand\tabref[1]{{Table \ref{tab:#1}}}

\AtBeginDocument{
  \abovedisplayskip     =0.5\abovedisplayskip
  \abovedisplayshortskip=0.5\abovedisplayshortskip
  \belowdisplayskip     =0.5\belowdisplayskip
  \belowdisplayshortskip=0.5\belowdisplayshortskip}

\begin{document}
\fancyhead{}

\title{Decomposition and Interleaving for \\Variance Reduction of Post-click Metrics}


\author{Kojiro Iizuka}
\affiliation{%
  \institution{Gunosy Inc. / University of Tsukuba}
}
\email{iizuka.kojiro@gmail.com}

\author{Yoshifumi Seki}
\affiliation{%
  \institution{Gunosy Inc.}
}
\email{yoshifumi.seki@gunosy.com}

\author{Makoto P. Kato}
\affiliation{%
  \institution{University of Tsukuba / JST, PRESTO}
}
\email{mpkato@acm.org}

\renewcommand{\shortauthors}{Trovato and Tobin, et al.}

\begin{abstract}
In this study, we propose an efficient method for comparing the post-click metric (e.g., dwell time and conversion rate) of multiple rankings in online experiments.
The proposed method involves (1) the decomposition of the post-click metric measurement of a ranking into a click model estimation and a post-click metric measurement of each item in the ranking,
and (2) interleaving of multiple rankings to produce a single ranking that preferentially exposes items possessing a high population variance.
The decomposition of the post-click metric measurement enables the free layout of items in a ranking and focuses on
the measurement of the post-click metric of each item in the multiple rankings.
The interleaving of multiple rankings reduces the sample variance of the items possessing a high population variance by optimizing a ranking to be presented to the users so that those items received more samples of the post-click metric.
In addition, we provide a proof that the proposed method leads to the minimization of the evaluation error in the ranking comparison
and propose two practical techniques to stabilize the online experiment.
We performed a comprehensive simulation experiment and a real service setting experiment.
The experimental results revealed that (1) the proposed method outperformed existing methods in terms of efficiency and accuracy, and the performance was especially remarkable when the input rankings shared many items,
and (2) the two stabilization techniques successfully improved the evaluation accuracy and efficiency.

\end{abstract}


\maketitle

\section{Introduction}
Online controlled experiments are conducted daily to evaluate search and recommendation algorithms.
A/B testing is a common approach that compares two different outcomes by showing them to two different user groups. 
A/B testing is applicable for a variety of purposes, including 
 measurement of click-based metrics (e.g., click-through rate (CTR)),
and measurement of {\it post-click metrics}
(e.g., music listening time~\cite{garcia2018understanding},
news reading time~\cite{okura2017yahoo, xing2014dwell}, and the number of reservations~\cite{grbovic2018airbnb}).
As post-click metrics are closely related to user satisfaction and the sales of services,
the measurement of post-click metrics is particularly important for the continuous improvement of algorithms in search and recommender systems.

However, there are some challenges in measuring post-click metrics in online experiments.
Suppose that a news article is ranked at the bottom of a ranking,
for which users spend a significantly different length of time to read.
Generally, low-ranked items are infrequently clicked.
Thus, the post-click metric results in high variance of the sample mean.
Moreover, some types of post-click metrics highly depending on users are of high {\it population} variance by nature, 
and also likely to result in a high variance of the sample mean.
High variance in each item naturally leads to high variance for the whole ranking.
This high variance can prevent online experiments from efficiently discriminating between competitive rankings.

This study proposes an efficient method for comparing the multiple rankings of post-click metrics in online experiments.
The key ideas of the proposed method are (1) the decomposition of the post-click metric measurement of a ranking into a click model estimation and post-click metric measurement of each item in the ranking,
and (2) interleaving multiple rankings to produce a single ranking that preferentially exposes items with high population variances.
The decomposition of the post-click metric measurement permits free layout of items in a ranking and focus on 
the measurement of the post-click metric of each item in multiple rankings.
The interleaving of multiple rankings reduces the variance of the sample mean for items with a high population variance by optimizing the ranking to be presented so that those items received more samples of the post-click metric.
In this paper, the method is referred to as the {\it Decomposition and Interleaving for  Reducing the Variance of post-click metrics} ({\it DIRV}).
DIRV uses interleaving to evaluate post-click metrics for the variance reduction, 
which is a major distinction from interleaving designed to evaluate click-based metrics. 

In addition to the proposal of DIRV, 
we have the following theoretical and technical contributions in this work.
First, we proved that the ranking optimization by the DIRV
 leads to minimization of the evaluation error in the ranking comparison.
Second, we propose two techniques to stabilize the evaluation by the DIRV.
The first technique predicts the population variance of each item based on the observed samples and item features
for stabilizing the ranking optimization at the beginning of the online experiment.
The second technique corrects the systematic error between the estimated post-click metrics and the ground-truth post-click metrics.

We performed comprehensive experiments in simulation as well as real service settings.
The experimental results revealed that (1) the proposed method outperformed existing methods in terms of efficiency and accuracy, and the performance was especially remarkable when the input rankings shared many items,
(2) the two stabilization techniques successfully improved the evaluation accuracy and efficiency.

The major contributions of this study are as follows:
\begin{itemize}
\item To efficiently compare post-click metrics of multiple rankings,
we proposed an interleaving method (DIRV) that decomposes the post-click metric measurement and preferentially exposes items with high population variance to minimize the evaluation error.
\item We provided a theoretical guarantee that the DIRV ranking optimization minimizes the evaluation error in the ranking comparison.
\item We proposed two techniques to stabilize the evaluation by DIRV and demonstrated that these techniques were empirically effective.
\item We extensively evaluated DIRV using both simulation and real service settings.
The results demonstrated its high accuracy and efficiency.
\item We published the real service data used in our experiments.
This data can be used to validate our study and be used in future research on user modeling and bandit algorithms.
\end{itemize}

The remainder of this paper is organized as follows: 
Section \ref{sec:relatedwork} provides an overview of the related work.
Section \ref{sec:background} describes the problem setting in this study.
Section \ref{sec:proposal} explains the proposed method.
Sections \ref{sec:experimentsetup} and \ref{sec:results} report the experimental settings and 
discuss the experimental results, respectively.
Finally, Section \ref{sec:conclusion} presents the conclusions and future work on this topic.

\section{Related Work}
\label{sec:relatedwork}
Online evaluations are the basis of data-driven decision-making~\cite{preexperiment}.
A typical online evaluation method is A/B testing~\cite{kohavi2013ab}, 
while there are also alternatives such as interleaving and bandit algorithms.

{\it Interleaving} (or {\it multileaving}, which refers to interleaving of more than two rankings) is a method used to increase the efficiency of online evaluation.
Interleaving was reported to be 10 to 100 times more efficient than A/B testing~\cite{radlinski2008does, chapelle2012large}.
Schuth et al. proposed two multileaving methods called team draft multileaving~(TDM) and optimized multileaving (OM)~\cite{schuth2014multileaved}.
Recently, sample-scored-only multileaving~(SOSM)~\cite{brost2016improved} and pairwise preference multileaving~(PPM) ~\cite{oosterhuis2017ppm} were developed as more scalable multileaving methods.
While PPM designed to evaluate click-based metrics is the state-of-the-art method in the interleaving, it is not trivial how to modify PPM to evaluate post-click metrics.
Specifically, Schuth et al. found that a simple extension of the TDM resulted in a low accuracy in the non-click-based metrics (e.g., time to click metrics)~\cite{schuth2015predicting}.
Our approach differs from previous interleaving methods because we formulated a new {\it credit function},
which is used to aggregate user feedback in the interleaving, 
as the expectation of post-click metrics for accurate evaluation.

Variance reduction is a technique commonly used to improve efficiency in online evaluations.
\cite{treemodel} used a boosted decision tree regression to reduce the variance by matching similar users.
The efficiency was improved in~\cite{preexperiment} by utilizing pre-experiment data through variance reduction.
\cite{casenetflix} implemented stratification for variance reduction.
These studies reduced the variance for each tested group in a bucket-level evaluation like A/B testing.
In contrast to bucket-level variance reduction, we reduce the item-level variances in post-click values based on the variance of the post-click values differing greatly for each item.
By reducing item-level variances, we were able to improve the evaluation efficiency.
Oosterhuis and de Rijke~\cite{oosterhuis2020taking} applied variance reduction to an efficient evaluation of CTR.
However, it was unclear how to apply this method to a post-click metric, which we focus on in this work.
The proposed method in~\cite{oosterhuis2020taking} relies heavily on the examination probability and
has only been validated by simulation-based experiments.
We found that the evaluation accuracy of the post-click metrics was highly compromised when there was an estimation error in the examination probability of the real service data.
We developed a stabilization technique to address this problem.
The stabilization technique reduces the systematic error caused by the estimation error of the examination probability.

Recently, several off-policy evaluation methods have been developed~\cite{gruson2019offline, joachims2017unbiased, wang2018position}.
The goal of off-policy evaluation is to estimate the performance of a policy from the data generated by another policy(ies) in reinforcement learning~\cite{farajtabar2018more}.
Specifically, Saito~\cite{saitoconversion} proposed an unbiased estimator for post-click metrics.
Our method is different in that we interleave rankings dynamically so that the estimation error is minimized by the variance reduction in an on-policy manner.

The multi-armed bandit problem is a problem in which a fixed limited set of resources must be allocated between competing choices in a way that maximizes the expected gain \cite{burtini2015survey, bouneffouf2019survey}.
Algorithms used for bandit problems are similar to the proposed method of evaluating multiple rankings in this study.
However, the objective of a bandit algorithm is different.
The objective of a bandit algorithm is to maximize the total gain during a specific period or to identify the best arm.
The objective of our study is to evaluate the superiority or inferiority of each pair of rankings.
The evaluation result for each pair cannot be obtained from a bandit algorithm.
However, the evaluation result for each pair may help decision-makers choose the best ranking in terms of effectiveness and the cost of producing the rankings.

\section{Problem Setting}
\label{sec:background}

The primary goal of this study is to compare multiple {\it input rankings} $R=\{r_1, r_2, \ldots\}$ in terms of a specific post-click metric in an online experiment.
A user is shown a ranking in response to her/his query,
clicks the items in the ranking, and consumes the clicked items.
The post-click metric of an item can only be measured after the item has been clicked by the user.
To compare rankings, we define the post-click metric of a ranking 
as the expectation of the post-click metric of items in ranking $r_i$:
\begin{eqnarray}
E[x|r_i] &=& \int_{x} xP(x|r_i) dx 
\label{original_expectation}
\end{eqnarray}
where $x$ is a random variable representing a value of the post-click metric,
and $P(x|r_i)$ is the probability of observing $x$ when a user interacts with the ranking $r_i$.
$E[x|r_i]$ can be interpreted as how effective the ranking is to trigger users' behaviors (e.g., conversion) quantified by the post-click metric.
Thus, $E[x|r_i]$ can be used to evaluate rankings based on users' post-click behaviors.

Following \cite{oosterhuis2017ppm, brost2016improved, schuth2015probabilistic}, 
we set the goal of our study to efficiently estimate the pairwise preference between rankings.
The pairwise preference between rankings is defined as a matrix ${\mathbf P} \in {\mathbb R}^{|R| \times |R|}$, 
where $P_{i,j}$ indicates the difference between the expected post-click metric for ranking $r_i$ and $r_j \in R$, 
that is, 
\[
P_{i,j} = E[x|{\it r_i}]  -  E[x|{\it r_j}].
\]
A positive value of $P_{i, j}$ indicates that the ranking $r_i$ is superior to $r_j$.
The binary error $E_{\rm bin}$ is commonly used error metric to evaluate $\mathbf{P}$~\cite{oosterhuis2017ppm, brost2016improved, schuth2015probabilistic}.
Letting $\overline{P}_{i,j}$ be the preference estimated by a certain method,
the binary error of this estimate is defined as follows:
\begin{equation}
E_{\rm bin} = \frac{1}{|R| (|R|-1)} \sum_{r_i, r_j \in R}{{\rm sgn}(P_{i,j}) \neq  {\rm sgn}(\overline{P}_{i,j})},
\label{binerror}
\end{equation}
where ${\rm sgn(\cdot)}$ returns -1 for negative values, 1 for positive values,
and 0 otherwise.
The operator $\neq$ returns 1 whenever the signs are unequal.

A straightforward approach for estimating $P_{i, j}$ (or obtaining $\overline{P}_{i, j}$) is to conduct A/B testing with the rankings $r_i$ and $r_j$, take the mean of $x$ for each of the rankings, and then approximate $E[x|{\it r_i}]$ and $E[x|{\it r_j}]$ using these means.
We discuss possible improvements over A/B testing and 
introduce our proposed method in the next section.

\section{METHODOLOGY}
\label{sec:proposal}
We propose a method named {\it Decomposition and Interleaving for Reducing the Variance of post-click metrics} ({\it DIRV}) for efficient evaluation based on post-click metrics.
We first discuss possible improvements over A/B testing by deeply look at the expectation of a post-click metric and introduce DIRV as our solution.
We then show that the optimization by DIRV leads to the reduction of the binary error defined in Equation (\ref{binerror}).
Finally, we propose two techniques to stabilize evaluations by DIRV in actual applications.

\subsection{Decomposition}
\label{clickmodelassumption}

We deeply look at Equation (\ref{original_expectation}) by decomposing it with some assumptions and discuss possible improvements to reduce the evaluation error.

In Equation (\ref{original_expectation}),
the value of a post-click metric $x$ can be observed for each clicked item in ranking $r_i$.
Thus, $P(x|r_i)$ can be obtained by marginalizing all of the items in $r_i$,
with the assumption that $x$ depends solely on $d$, but not on the ranking $r_i$:
\begin{eqnarray}
P(x|r_i) = \sum_{d \in r_i} P(x|d)P(c_d=1|r_i) 
\end{eqnarray}
where $c_d$ is a random binary variable indicating an event of click on $d$,
$P(c_d=1|r_i)$ is the click probability of $d$ in the ranking $r_i$,
and $P(x|d)$ is the probability of observing $x$ at $d$.

The independence assumption between $x$ and $r_i$
allows the following decomposition of the expectation of the post-click metric:
\begin{eqnarray}
E[x|r_i] &=& \int_{x} x \left\{ \sum_{d \in r_i} P(x|d)P(c_d=1|r_i) \right\} dx \nonumber \\ 
&=& \sum_{d \in r_i} \int_{x} x  P(x|d) P(c_d=1|r_i) dx \nonumber \\
&=& \sum_{d \in r_i} P(c_d=1|r_i) \int_{x} x  P(x|d)  dx 
= \sum_{d \in r_i} P(c_d=1|r_i) E[x|d] \nonumber \\
\label{expectation}
\end{eqnarray}
where $E[x|d]$ is the expectation of the post-click metric for item $d$. 

This equation immediately suggests that an unbiased estimate of $E[x|r_i]$ can be obtained by:
\begin{eqnarray}
\overline{E}[x|{\it r_i}] = \sum_{d \in r_i} \overline{P}(c_d=1|r_i) \overline{E}[x|d] 
\label{eq:first_decomposition}
\end{eqnarray}
where $\overline{P}(c_d=1|r_i)$ and $\overline{E}[x|d]$ are unbiased estimates for $P(c_d=1|r_i)$ and $E[x|d]$, respectively.
A standard approach to obtain these two estimates is to use the sample means of the random variables $c_d$ and $x$ for $\overline{P}(c_d=1|r_i)$ and $\overline{E}[x|d]$, respectively.
This decomposition can potentially achieve higher efficiency than A/B testing,
because the samples for item $d$ in {\it any} rankings in $R$
can be used to estimate $E[x|d]$.
Thus, we can increase the sample size for items that are shared by multiple rankings.
Rankings are likely to include identical items, especially when differences between similar rankings (e.g., the same algorithm with different parameters) are evaluated.

While only Equation (\ref{eq:first_decomposition}) enables better estimation of $E[x|d]$ by increasing the sample size, 
the use of the sample mean is not efficient enough especially when 
(1) the population variance of $x$ is high and 
(2) the sample size of $x$ is small due to small $P(c_d=1|r_i)$.
Both cases are likely to lead to high variance of  $\overline{E}[x|d]$, resulting in high variance of $\overline{E}[x|{\it r_i}]$.
Consequently, a large binary error in the ranking comparison is obtained.
Both cases can be alleviated by prioritizing 
the exposure of items that satisfy these conditions for increasing the sample size.
However, simple manipulation of a ranking prevents us from correctly estimating $P(c_d=1|r_i)$
since this probability depends on the position of item $d$ in ranking $r_i$.

Hence, we introduce a click model for further decomposing $\overline{P}(c_d=1|r_i) \overline{E}[x|d]$ in the summation.
We assume the examination hypothesis~\cite{chuklin2015click} that the item click can be decomposed into two variables: examination $e_d$ and attraction $a_d$.
This enables us to decompose the unbiased estimate $\overline{P}(c_d=1|r_i)$
as follows:
\begin{eqnarray}
\overline{P}(c_d=1|r_i) = \overline{P}(e_d = 1|r_i) \overline{P}(a_d = 1| d),
\label{eq:clickmodel}
\end{eqnarray}
where $\overline{P}(e_d = 1|r_i)$ can be estimated by a position-based click model or a cascade click model~\cite{chuklin2015click, craswell2008experimental}.
The position-based click model assumes that the examination probability depends solely on the rank of the item in ranking $r$: i.e., $\overline{P}(e_d = 1|r_i) = g({\rm rank}(d, r))$, where $g$ is a function taking a rank to return a probability, and ${\rm rank}(d, {\it r})$ is the rank of item $d$ in ranking $r$.
This assumption allows us to avoid estimating the probabilities specific to a particular ranking: i.e., $\overline{P}(e_d = 1|r_i)$.
The estimations of $E[x|d]$, ${P}(e_d = 1|r_i)$ and $P(a_d = 1|r_i)$ are detailed in Section \ref{parameterestimation}.
Note that the simple assumption for the click model might sacrifice the evaluation error for efficiency.
Thus, in Section \ref{sq:techniques}, we introduce a technique correcting the systematic error between the assumed and actual click models.

In summary, the decomposition of Equation (\ref{original_expectation})
increases the sample size of $x$ for estimating $E[x|d]$ and provides greater flexibility about the ranking presented to the users.
Now one can freely produce and present a ranking containing items from input rankings $R$,
and estimate each of $P(a_d = 1| d)$ and $E[x|d]$ based on the users' interactions with the presented ranking, in order to obtain $\overline{E}[x|{\it r_i}]$.
In the next subsection, we explain what ranking should be presented to the users for minimizing the evaluation error.

\subsection{Interleaving}

Our proposed DIRV method is designed to minimize the variance of a post-click metric 
by interleaving input rankings to produce a single ranking to be presented to the users.
The interleaved ranking preferentially exposes items with high population variance.
The method attempts to receive more samples (or clicks) for these items to reduce the variance of the input rankings.

As shown in Appendix \ref{dep},
the variance $V[\overline{E}[x|r_i]]$ can be defined 
as a monotonically decreasing function $\phi$ on the number of impressions of the item $d$ (denoted by $n^i_d$) and the number of clicks on the item $d$ (denoted by $n^c_d$):
\begin{eqnarray}
V[\overline{E}[x|r_i]] = \sum_{d \in r_i} \phi_{d, r_i}(n^i_d, n^c_d).
\label{varianceeqr}
\end{eqnarray}

Our proposed DIRV method
determines a ranking $o$ produced by interleaving input rankings.
The ranking $o$ minimizes the summation of $V[\overline{E}[x|{\it r_i}]]$ over all the rankings in $R$ 
when $o$ is exposed to a user:
\begin{eqnarray}
\min_o f(o) \nonumber
\label{ranking_optimization}
\end{eqnarray}
where 
\begin{eqnarray}
f(o) = \sum_{r_i \in R} V[\overline{E}[x|r_i]] = \sum_{r_i \in R} \sum_{d \in r_i} \phi_{d, r_i}(\dot{n}^i_d, \dot{n}^c_d ), \\
\dot{n}^i_d = n^i_d + 1, \ \ \ \ \dot{n}^c_d = n^c_d + E[c_d|o].
\end{eqnarray}
$\dot{n}_d$ indicates the expected sample size 
after the interleaved ranking $o$ is presented to a user.
The number of impressions, $n^i_d$, is incremented by one each time the ranking $o$ is presented.
The number of clicks for item $d$, $n^c_d$, is incremented by $E[c_d|o] = P(c_d=1|o)$,
which indicates the expected number of clicks on item $d$ in ranking $o$ for a single impression.

As a brute-force search is infeasible to determine the optimal ranking $o$,
we employ a greedy algorithm based on Equation (\ref{ranking_optimization}).
The greedy algorithm starts with an empty ranking $r$,
and repeats appending an item to $r$ that maximizes the difference between the current ranking and the ranking with the new item in terms of the variance of the sample mean:
\begin{equation}
\max_{d \in D \setminus r} \sum_{r_i \in R} \left(
\phi_{d, r_i}(n^i_d, n^c_d) - 
\phi_{d, r_i}(\ddot{n}^i_d, \ddot{n}^c_d) \right),
\end{equation}
where
\begin{equation}
\ddot{n}^i_d = n^i_d + 1, \ \ \ \ \ddot{n}^c_d = n^c_d + E[c_d|r \oplus d].
\end{equation}
and $D$ is a set of all items in $R$ and $r \oplus d$ is the ranking $r$ with $d$ appended to the bottom.
This greedy algorithm repeatedly finds the item to minimize the variance when it is appended.
The algorithm stops when the ranking reaches a predefined depth.

\subsection{Theoretical Justification}

We provide theoretical justification that the minimization of the summation of $V[\overline{E}[x|{\it r_i}]]$ over $R$ leads to the minimization of the expected binary error $E[E_{\rm bin}]$ defined in Equation (\ref{binerror}). 

\begin{theorem}
\begin{eqnarray}
\sum_{r_i \in R}V[\overline{E}[x|r_i]]
\end{eqnarray}
constitutes the upper bound of $E[E_{\rm bin}]$.
\label{biaserror}
\end{theorem}
\begin{proof}
Let $\mu_i = E[x|r_i]$ and $\overline{\mu}_i = \overline{E}[x|r_i]$.
Without a loss of generality, we assume that $\mu_i > \mu_j$.
The binary error probability $P({\rm sgn}(P_{i,j}) \neq  {\rm sgn}(\overline{P}_{i,j}))$
can be interpreted as the probability of estimating a higher value for $\mu_j$ than that for $\mu_i$,
that is, $P({\rm sgn}(P_{i,j}) \neq  {\rm sgn}(\overline{P}_{i,j}))
=P(\overline{\mu}_j > \overline{\mu}_i) = P(\overline{\mu}_j - \overline{\mu}_i > 0)$.
Letting $\Delta_{ij}=\mu_j - \mu_i$ and $\overline{\Delta}_{ij} = \overline{\mu}_j - \overline{\mu}_i$,
and using Chebyshev's inequality ($P(|x - \mu| > k) \leq \sigma^2 / k^2$), 
this probability can be bounded as follows:
\begin{eqnarray}
P({\rm sgn}(P_{i,j}) \neq  {\rm sgn}(\overline{P}_{i,j})) \nonumber 
&=& P(\overline{\Delta}_{ij} - \Delta_{ij} > - \Delta_{ij})  \nonumber \\
&\leq& P(|\overline{\Delta}_{ij} - \Delta_{ij}| > - \Delta_{ij}) \nonumber \\
&\leq& \frac{V[\overline{\Delta}_{ij}]}{\Delta_{ij}^2} 
= \frac{V[\overline{\mu}_i] + V[\overline{\mu}_j]}{\Delta_{ij}^2}
\end{eqnarray}
Note that $P(|x|>y)=P(x>y)+P(x<-y) \geq P(x>y)$,
and $V[x-y] = V[x] + V[y]$ if $x$ and $y$ are assumed independent.

Therefore, the expected binary error is bounded as follows:
\begin{eqnarray}
E[E_{\rm bin}]&=&\frac{1}{|R|(|R|-1)}\sum_{{r_i,r_j} \in R} P({\rm sgn}(P_{i,j}) \neq  {\rm sgn}(\overline{P}_{i,j})) \nonumber \\
&\leq& \frac{1}{|R|(|R|-1)}\sum_{{r_i,r_j} \in R} \frac{V[\overline{\mu}_i] + V[\overline{\mu}_j]}{\Delta_{ij}^2} \nonumber \\
&\leq& \frac{1}{C|R|}\sum_{{r_i} \in R} V[\overline{\mu}_i] 
= \frac{1}{C|R|} \sum_{r_i \in R}V[\overline{E}[x|r_i]]
\end{eqnarray}
where $C$ is the smallest value for $\Delta_{ij}^2$ ($r_i, r_j \in R$).
Note that $V[\overline{\mu}_i]$ appears $|R| - 1$ times in the summation,
which cancels out $|R| - 1$ in the denominator.
\end{proof}

This theorem suggests that the upper bound of $E[E_{\rm bin}]$ becomes lower if the variance of $\overline{E}[x|r_i]$ is reduced.
This is the theoretical basis of our proposal
and is formally expressed as follows:
\begin{eqnarray}
E[E_{\rm bin}] \rightarrow 0 \ \ \left(\sum_{r_i \in R}V[\overline{E}[x|r_i]] \rightarrow 0 \right).
\end{eqnarray}

\subsection{Stabilization Techniques}
\label{sq:techniques}

There still remain some practical problems associated with the application of DIRV to an online experiment.
These problems are discussed in this section.

\subsubsection{Feature-based Variance Prediction}
The first problem is the variance estimation of the post-click metric $V[x|d]$.
This variance is included in the variance of $\overline{E}[x|d]$ (see Appendix \ref{dep})
and is required to be estimated to produce interleaved rankings. 
However, the estimation can be inaccurate, especially when the sample size is small.
As demonstrated in our experiments, this problem can cause degradation of the estimation efficiency.

To overcome this problem, we propose using a regression model that predicts the variance of the post-click metric based on item features (e.g., titles and categories). 
As a result, we can use two estimates for the variance of the post-click metric $V[x|d]$,
namely, the unbiased variance obtained from the observed values (denoted by $\overline{V}[x|d]$),
and the predicted variance estimated by a regression model (denoted by $\hat{V}[x|d]$).

The larger value between $\overline{V}[x|d]$ and $\hat{V}[x|d]$
(i.e., $\max(\overline{V}[x|d], \hat{V}[x|d])$), is used as the estimate of $V[x|d]$.
This idea is similar to the {\it clipping technique} used in~\cite{chen2019top, gilotte2018offline, su2019cab}.
As DIRV prioritizes items with a high variance in the interleaved ranking, 
there are opportunities to correct the estimation errors 
for items whose variance of the post-click are overestimated.
In contrast, underestimation of the variance is more problematic because there are only limited opportunities to correct the error.
Therefore, we use a higher variance to avoid underestimating the variance.
Our experiments demonstrated that this technique is effective for efficient evaluation.

\subsubsection{Systematic Error Correction}
\label{biasreduction}
The second technique reduces the systematic error between the actual and assumed click models.
Although the click probability $P(c_d=1|r_i)$ is efficiently estimated with the click model,
its estimation accuracy may not be sufficient as the click model does not precisely reflect real user behaviors.

To alleviate this problem, we estimate the click probability $P(c_d=1|r_i)$ by the click model as well as a model-agnostic estimate (i.e., the sample mean of clicks).
We present each ranking $r_i$ with a small probability and observe the clicks on $d$ in $r_i$.
The clickthrough rate (i.e., $\hat{P}(c_d=1|r_i) = n^c_{d, r_i} / n^i_{r_i}$)
is used together with $\overline{P}(c_d=1|r_i)$ in Equation (\ref{eq:clickmodel}),
where $n^c_{d, r_i}$ is the number of clicks on item $d$ in the ranking $r_i$
and $n^i_{r_i}$ is the number of impressions of the ranking $r_i$.
The two estimates are linearly combined for approximating $P(c_d=1|r_i)$:
\begin{eqnarray}
P(c_d=1|r_i) \approx \theta_i \overline{P}(c_d=1|r_i) +  (1 - \theta_i) \hat{P}(c_d=1|r_i)
\end{eqnarray}
where we set the parameter $\theta_i$ to a value decreasing as the sample size increases ( i.e., $1/(n^i_{r_i} + 1)^{0.5}$) to weight the click model when $n^i_{r_i}$ is small.

To present each ranking $r_i$ with a small probability, 
we incorporate another objective function $g(o)$ into the policy for generating rankings.
$g(o)$ is defined as follows:
\begin{eqnarray}
g(o) = \sum_{{\it r_i} \in {\it R}} \theta_i \sum_{d \in {\it r_i}} 
\phi_{d, r_i}(\dot{n}^i_{r_i}, \dot{n}^c_{d, r_i})
\label{final_judgement}
\end{eqnarray}
where 
\begin{eqnarray}
\dot{n}^i_{r_i} = n^i_{r_i} + {\rm sgn}(r_i=o),  \ \ \ \ \dot{n}^c_{d, r_i} = n^c_{d, r_i} + {\rm sgn}(r_i=o)P(c_d=1| o)
\end{eqnarray}
and ${\rm sgn}(r_i=o)$ returns $1$ if all of the items are identical in $r_i$ and $o$, otherwise $0$.
In the last procedure for generating rankings, we select the ranking $o$ that minimizes $f(o) + \gamma g(o)$ from the set $\{o^* \}\cup \it{R}$, where $o^*$ is a greedy solution in Equation (\ref{ranking_optimization}), and $\gamma$ is a hyper-parameter.
Although this technique is not justified theoretically, 
it was shown to be effective by our experiments.

\section{Experiment Setup}
\label{sec:experimentsetup}
We describe the experimental settings to answer the following research questions (RQs):
\begin{itemize}
\label{researchquestion}
\item {\bf RQ1}: Can DIRV identify preferences
between rankings more accurately and efficiently than other methods?
\item {\bf RQ2}: How does the variance prediction technique affect the evaluation efficiency?
\item {\bf RQ3}: How does the error correction technique affect the evaluation accuracy?
\end{itemize}

We have two experimental settings for comprehensive validation, namely, the simulation-based setting and the real service setting.
A summary of the experimental settings is included in \tabref{experimetalsummary}.
The simulation-based setting is based on the traditional setting but is extended for the post-click evaluation.
The real service setting is designed to validate methods close to the actual online evaluation by use of raw user feedback about the clicks and post-click values obtained from the raw ranking impressions.
The real service dataset\footnote{It is public available at  https://github.com/koiizukag/DIRV.} are described in Section \ref{realworlddata}.

\subsection{Simulation-based Settings}
\subsubsection{Evaluation Procedure}
With the exception of post-click behavior, we simulated user behavior using the same four steps as in the previous interleaving evaluations \cite{radlinski2013optimized, schuth2014multileaved}:
\begin{enumerate}
    \item {\bf Impression}: A list of ranked items is displayed for a user.
    \item {\bf Click}: The user decides whether to click an item.
    \item {\bf Post-click Behavior}: If the user clicks the item, then the user consumes it, and post-click values are produced.
    \item {\bf Session End}: The user ends the session.
\end{enumerate}
The details of the user's behaviors are as follows:

First, we fix a query by uniformly sampling from a static dataset. 
Then, each method generates rankings from the query and displays them to the users. 
Whether a user clicks on an item depends on the item's click probability.

\begin{table}[t]
  \caption{Experimental Setting Summary}
  \label{tab:experimetalsummary}
  \centering
  \begin{tabular}{ccc}
    \hline
      & Simulation-based  & Real Service \\
    \hline \hline
    Input ranking  & generated & raw data \\
    Systematic error & - & exists \\
    Click & partially generated & raw data \\
    Post-click & partially generated & raw data \\
    \hline
  \end{tabular}
\end{table} 

After clicking on the item, the user is assumed to consume the item.
Then, we can measure the value of a post-click metric (e.g., dwell time).
The post-click value is assumed to follow a particular distribution.
For example, the occurrence of a conversion in EC dataset follows a Bernoulli distribution.
Throughout the experiment, we assumed a cascade click model~\cite{chuklin2015click},
in which users view results from top to bottom and leave as soon as they see a worthwhile item.
Each experiment involved a sequence of 10,000 simulated user impressions. 
All runs were repeated 30 times for each dataset and parameter setting.

\subsubsection{Datasets}
Three datasets were used in the simulation-based settings.
The first dataset (called News) was generated from our news media service named Gunosy. 
The second dataset (called LETOR) is the learning to rank dataset.
These two datasets were used to measure the dwell time that varies for each user.
The third dataset (called EC) is an artificially generated dataset used to simulate e-commerce product purchases that has a predefined post-click value for the purchase price of the product.

\paragraph{\bf News Dataset}
This dataset comes from a real-world, mobile news application.
The News dataset contains the dwell time and the click probability.
We generated this dataset using the implicit feedback from a single day.
Each selected article had more than 20,000 raw post-click values (i.e., the dwell time for each user).
This dataset contains 50 articles, which is sufficient to generate rankings.
The ground truth for click probability and dwell time for each article was calculated from all of the user logs. 

\paragraph{\bf LETOR Dataset}
We used eight publicly available LETOR datasets with varying sizes and representing different search tasks~\cite{schuth2014multileaved, oosterhuis2017ppm}.
Each dataset consisted of a set of queries and a set of items corresponding to each query.
Feature representations and relevance labels were provided for each item and query pair.
The datasets had three labels: unrelated (0), relevant (1), and highly relevant (2).
The LETOR dataset was broken down by task.
Most of the tasks were from the TREC Web Tracks from $2003$ to $2008$ \cite{qin2016letor, clarke2009overview, voorhees2003overview}. 
These TREC Web Tracks were HP2003, HP2004, NP2003, NP2004, TD2003, and TD2004.
Each TREC Web Track contained 50--150 queries.
The OHSUMED dataset was based on a search engine's query log of the MEDLINE abstract database and contained 106 queries.
MQ2007 was based on the million query track~\cite{allan2007million} and consisted of 1,700 queries.

We generated click probability and post-click values because these datasets do not contain these values.
We assumed that a user would read an article for a long time if the article was relevant to the user.
To reflect this, we generated dwell time values using an exponential distribution with the rate parameters set to $\frac{1}{\lambda_d} \sim (\text{relevance}+1) \cdot \text{uniform}(1, 20)$ for each item.
We set the click probability to $P(a_d = 1| d) \sim \text{min}((\text{relevance}+1) \cdot \text{uniform~}(0.0, 0.5), 1)$.

\paragraph{\bf EC Dataset}
We generated a dataset to simulate a product purchase in an e-commerce setting.
This dataset can be easily reproduced based on the parameters described below.
In this dataset, we randomly set the click probability to $P(a_d = 1| d) \sim \text{uniform~}(0.0, 0.5)$, $\text{price}(d) \sim \text{uniform~}(1, 1000)$, and conversion\_rate$(d) \sim \text{uniform~}(0, 0.5)$.
The relationship between the conversion rate and the price is $E[x|d]=\text{conversion\_rate}(d) \cdot \text{price}(d)$.
We generated 50 items using these parameters.

\subsubsection{Input Ranking}
In the LETOR datasets, we generated {\it input rankings} by sorting items with features used in the existing interleaving experiments~\cite{schuth2014multileaved, oosterhuis2017ppm}.
We used the BM25, TF.IDF, TF, IDF, and LMIR.JM features for MQ2007.
For the other datasets, we used the BM25, TF.IDF, LMIR.JM, and Hyperlink features.
First, we randomly selected 20 items to obtain the same number of item candidates for each dataset.
We then sorted ten items by each feature.
As a result, we generated $|R|=5$ rankings of length $|r_i|=10$, which were similar to the settings reported in \cite{schuth2014multileaved}.

For the News and EC datasets, we generated input rankings by prioritizing items with positive feedback.
We assumed that News and EC companies are likely to pay more attention to the user's feedback in their ranking algorithm.
The algorithm for generating input rankings first retrieved the top-$k$ items ranked in descending order of $P(a_d = 1| d)E[x|d]$.
This algorithm was used to approximate the level of user satisfaction.
Next, we randomly retrieved items that had not been selected and appended them to the input ranking.
We continued this process until the ranking reached a certain length.
Finally, we randomly shuffled the order of each input ranking randomly.
We used different values for $k$ in our experiments, which we called the {\it item duplication rate}.
As a result, we generated $|R|=5$ input rankings of length $|r_i|=10$, which are similar values to the setting reported by~\cite{schuth2014multileaved}.

\begin{center}
\begin{table*}[t]
\caption{\bf Evaluation accuracy for the News, EC, and LETOR datasets in the simulation-based setting. The binary error $E_{\rm bin}$ of all the comparison methods after 10,000 impressions on comparisons of $|R| = 5$ rankings was averaged over 30 times per dataset and parameter. The best performances are noted in bold.}

\begin{tabular}{l|ccccc|ccccc} \hline  \hline 
& \multicolumn{5}{c}{News} & \multicolumn{5}{c}{EC} \\ \hline 
Duplication ratio (\%) &0&20&40&60&80&0&20&40&60&80 \\ \hline
A/B Testing&0.100&0.190&0.095&0.140&0.160&                                           \bf 0.015&0.040&0.070&0.075&0.185 \\ \hline
TDM&0.425&0.465&0.470&0.500&0.615&                                                       0.055&0.170&0.200&0.215&0.300 \\ \hline
DIRV w/o Var Pred&0.150&0.280&0.115&0.100&0.110&                                         0.350&0.280&0.275&0.320&0.335 \\ \hline
DIRV w/o Err Corr&\bf 0.075&\bf 0.095&\bf 0.070&\bf 0.035&\bf 0.070&                     0.035 & 0.045&0.055&\bf 0.030&\bf 0.050 \\ \hline
DIRV&0.095&0.165&0.080&0.045&0.075&                                                  \bf 0.015&\bf 0.035&\bf 0.045&0.055&\bf 0.050 \\ \hline
\end{tabular}
\begin{tabular}{l|ccccccccc} \hline
& \multicolumn{8}{c}{LETOR} \\ \hline 
&HP2003&HP2004&TD2003&TD2004&NP2003&NP2004&MQ2007&OHSUMED \\ \hline
A/B Testing&0.085&0.105&0.070&0.110&0.095&0.080&0.100&0.085 \\ \hline
TDM&0.360&0.255&0.295&0.345&0.285&0.330&0.355&0.295 \\ \hline
DIRV w/o Var Pred&0.070&0.055&0.060&0.060&0.075&0.055&0.065&0.035 \\ \hline
DIRV w/o Err Corr&\bf 0.020&0.045&\bf 0.035&\bf 0.025&\bf 0.030&\bf 0.030&\bf 0.015&0.040 \\ \hline
DIRV&0.030&\bf 0.015&\bf 0.035&0.040&0.045&0.035&0.055&\bf 0.020 \\ \hline
\end{tabular}

\label{tab:resaccuracy}
\end{table*}
\end{center}

\subsection{Real Service Settings}
\label{realworlddata}
Simulation-based experiments are easily conducted.
However, their settings are limited in the validation of their assumptions.
In the proposed method, estimating the expected post-click metric relies on the assumption that post-click behavior is independent of other behaviors by the user.
However, this assumption may not always be true.
For example, in the news service, the dwell time on a clicked article may depend on the dwell time of another clicked article when the two articles contain similar or redundant content.
The simulation-based experiments were based on user simulation, which assumes the post-click behavior is independent.
Moreover, the assumed click model may have been wrong or contained estimation errors.
Therefore, the validity of these assumptions was not tested by the simulation-based experiments.
We generated a real service dataset using the rankings presented to real users
to validate the proposed method under conditions closer to actual user behaviors in an online evaluation.

\subsubsection{Real Service Dataset}
The evaluation of interleaving requires 1) {\it input rankings} to be used as the input, 2) generation of {\it interleaved rankings} from the {\it input rankings} for presentation to the users, and 3) {\it user behaviors} associated with the interleaved rankings.
The interleaving method generates interleaved rankings from the input rankings using various policies for generating the rankings.
To evaluate methods in this setting, we generated interleaved rankings by identifying {\it all possible combinations} of the items that the input rankings could generate and collected user behavior associated with the interleaved rankings from the service log.
As they cover all possible interleaved rankings, 
one can obtain real statistics observed for any interleaved rankings without actually presenting them.

The data were collected in a news application.
The news application was the same application used in the simulation-based News dataset.
This dataset consisted of a set of queries, a set of input rankings and the interleaved rankings with user behaviors corresponding to each query.
User behaviors data contained click or not for each article in the ranking and dwell-time for each article that the user clicked. 
The query was composed of the topic and day requested by the user.
We used a social topic that was one of the most popular topics in this service.
The ranking was personalized and differed for each user and query pair.
The top three rankings with the highest number of impressions for each day were selected as the input ranking.
The rankings were fixed at three in length to reduce the possible number of item combinations,
since the number of interleaved rankings could be explosively large depending on the number of items in the input rankings.
We generated 30 pairs of queries and their respective responses for a million-scale ranking impression in a one-month period.

The real service setting has advantages and disadvantages over simulation-based settings.
The real service dataset was limited by the number of items in the input rankings.
In addition, interleaved rankings did not cover all combinations of the items in the input rankings.
On the other hand, 
this real service setting did allow us to test the assumptions of the model close to the actual online evaluation.

\subsubsection{Experimental Runs}
The ground truth of the preference in $E_{\rm bin}$ was calculated from half of the total data included for each query using A/B testing logic.
The method was validated using the other half of the data.
The experimental runs consisted of 5,000 ranking impressions for each of the 30 queries.
These procedures were iterated 30 times.

\begin{figure*}[t]
\begin{center}
\begin{tabular}{ccc}
\begin{minipage}{0.3\hsize}
\begin{center}
\includegraphics[scale=0.3]{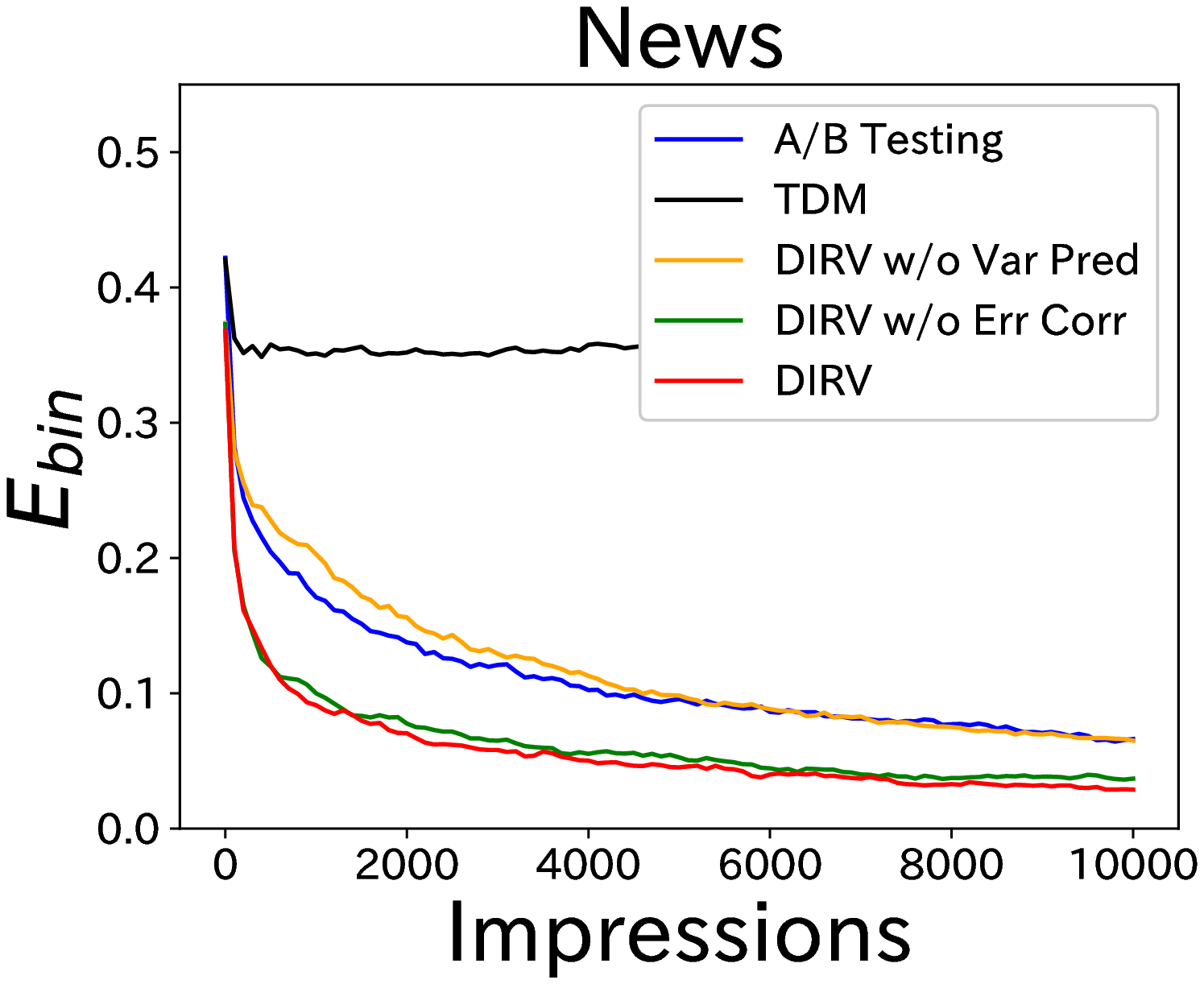}
\end{center}
\end{minipage}
&
\begin{minipage}{0.3\hsize}
\begin{center}
\includegraphics[scale=0.3]{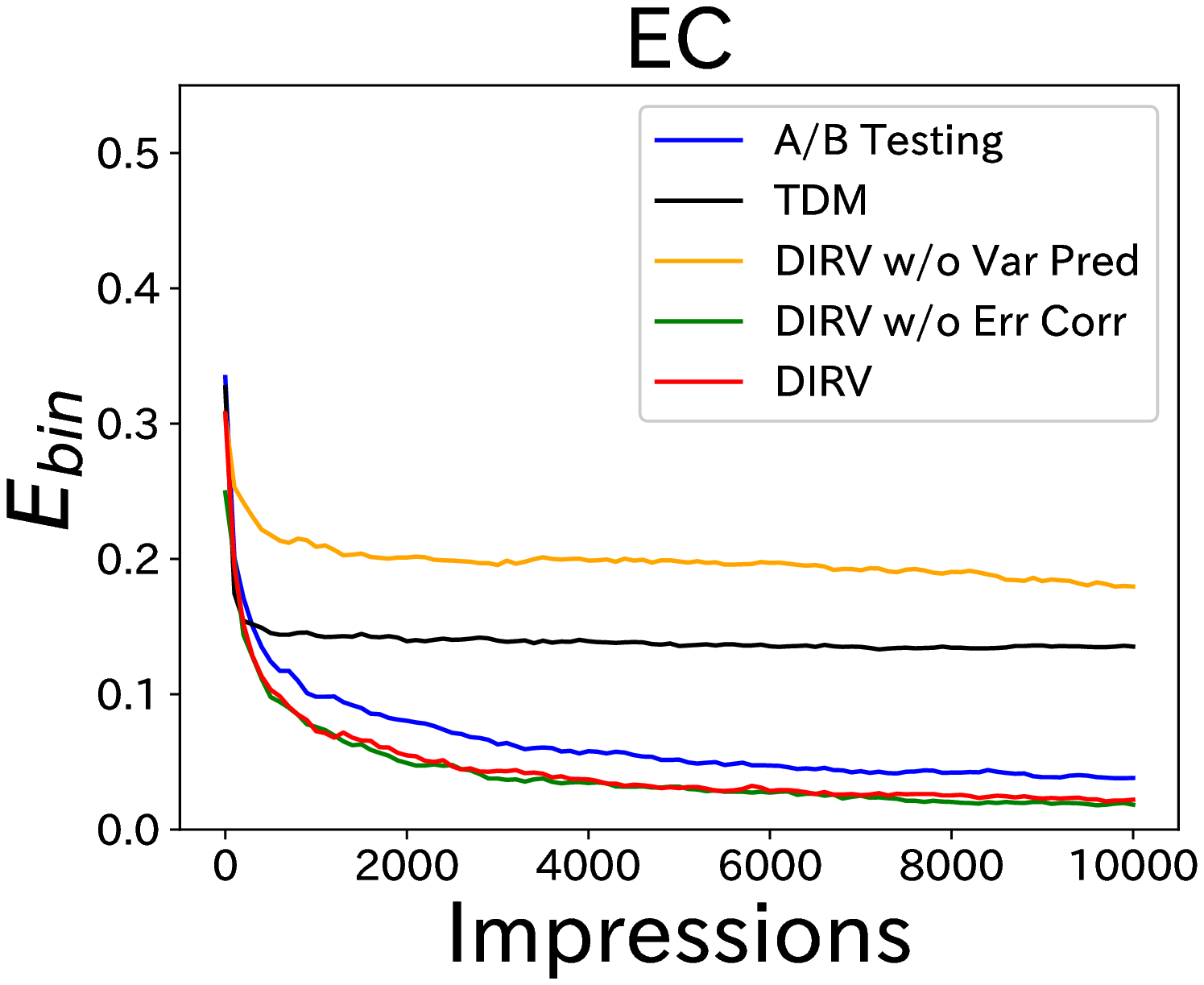}
\end{center}
\end{minipage}
&
\begin{minipage}{0.3\hsize}
\begin{center}
\includegraphics[scale=0.3]{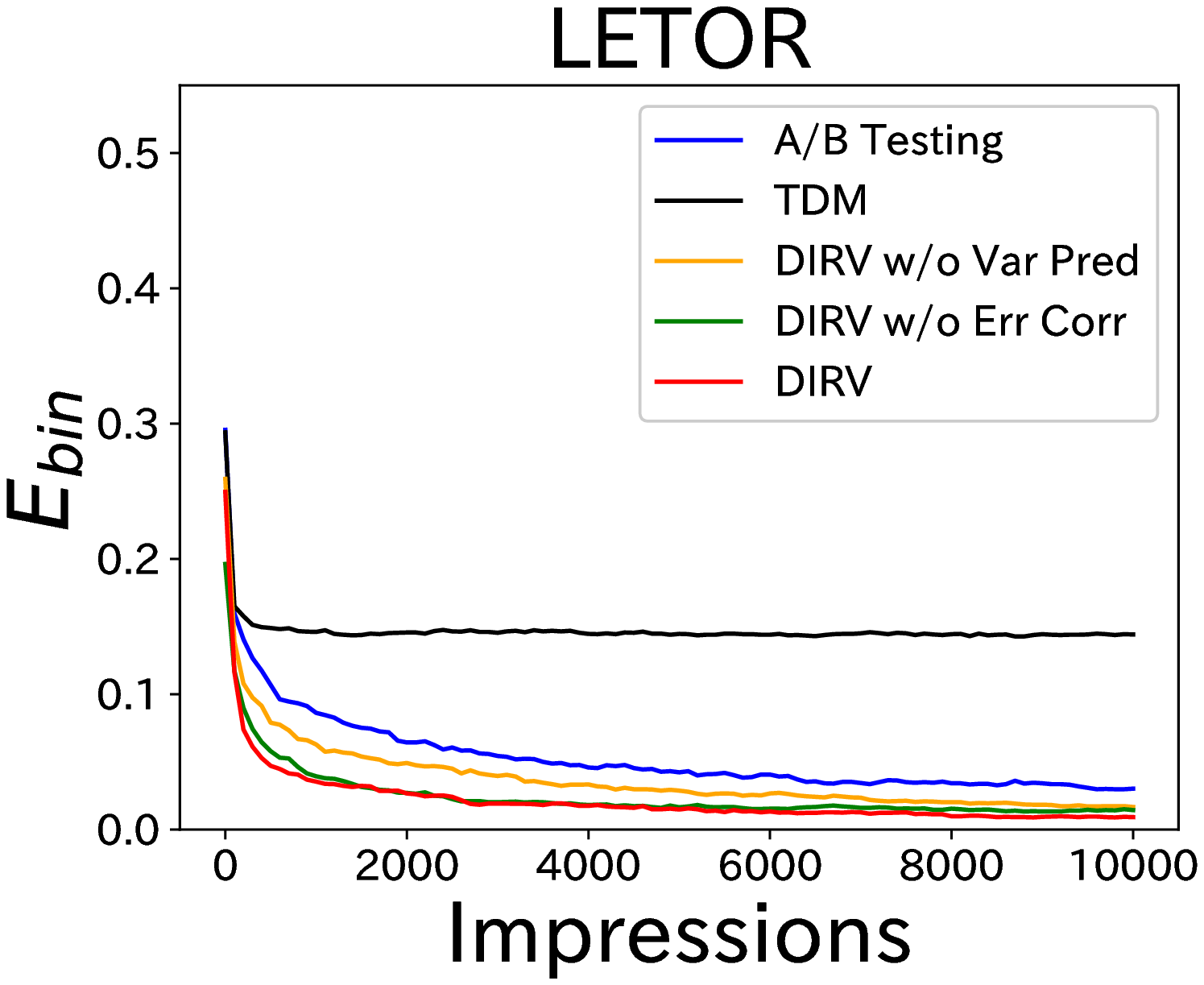}
\end{center}
\end{minipage}
\end{tabular}
\caption{$E_{\rm bin}$ averaged over the number of impressions of the News, EC, and LETOR datasets. For all the datasets, DIRV or DIRV w/o Err Corr had the lowest $E_{\rm bin}$ for each impression.}
\label{fig:efficiencty}
\end{center}
\end{figure*}

\subsection{Parameter Estimation}
\label{parameterestimation}
Each parameter was estimated as follows:
$E[x|d]$ was estimated by the sample mean of the observed post-click value $x$ for item $d$ over all rankings.
${P}(e_{d_j}=1|r_i)$ that was the examination probability of $j$-th item in the ranking $r_i$ was estimated assuming the cascade click model as $\overline{P}(e_{d_j}=1|r_i) = \prod_{k=1}^{j-1} (1-\overline{P}(c_{d_k}=1|d_k))$ where
$\overline{P}(c_{d_k}=1|d_k)=\overline{P}(e_{d_k}=1|d_k)\overline{P}(a_{d_k}=1|d_k)$.
$P(a_d=1|d)$ was estimated by $\overline{P}(a_d=1|d)=n^c_d/n^e_d$ where $n^c_d$ is the number of clicks and $n^e_d$ is the number of presentations of item $d$ over all of the rankings.
We note that $n^e_d$ is incremented for the items between the 1st position and the last position of the click in the ranking for each ranking impression.
If there were no clicks in the ranking, $n^e_d$ for each item was increased. 

\subsection{Variance Prediction}
We used the News and real service datasets for studying the variance prediction because the EC and LETOR datasets had no ground-truth for the variance.
We generated a dataset that included each article's title length and category, as well as the media source and the sample variance for ground-truth that was calculated from the dwell time of all users.
A tree-based training model with a gradient-boosting framework~\footnote{https://github.com/microsoft/LightGBM} was used for prediction.
We used the features detailed in \tabref{featureimportance} for prediction.
The training epoch was set to 1,000.
The root mean square error was used for the loss function.
The early stopping parameter was 10. 
The other parameters were set to default values.

\begin{table}[tb]
 \caption{Features and importance}
 \label{tab:featureimportance}
 \centering
  \begin{tabular}{lc}
   \hline
   Feature & Importance \\
   \hline \hline
   Category ID to which the article belongs & 879 \\
   Supplier ID of the article & 2,342  \\
   Content length of the article & 1,854 \\
   Title length of the article  & 1,045 \\
   \hline
  \end{tabular}
\end{table}
\figref{varianceprediction} shows the plots of the predicted variance and the actual variance.
A Pearson correlation value of 0.76 was achieved by only using meta-features.
Feature importance is shown in \tabref{featureimportance}.

In the EC and LETOR datasets, we artificially generated the predicted variance. 
Specifically, uniformly randomized noise was added to the population variance $V[x|d]$ to generate a predicted $\hat{V}[x|d]$ that satisfies $\hat{V}[x|d] \leq 2V[x|d]$.

\subsection{Comparison Methods}
We used five methods for the evaluation: A/B Testing, modified TDM, DIRV w/o Var Pred, DIRV w/o Err Corr, and DIRV.
PPM is designed to evaluate click-based metrics and considered as the state-of-the-art method for interleaving~\cite{oosterhuis2017ppm}.
However, it is not trivial to modify PPM to evaluate post-click metrics.
As previously discussed, the modified TDM~\cite{schuth2015predicting} is the only method designed to evaluate non-click-based metrics using interleaving.
Based on~\cite{schuth2015predicting}, we modified the credit function to be the product of the post-click value and the original credit function of the TDM.
DIRV w/o Var Pred is a DIRV method without the variance prediction,
while DIRV w/o Err Corr is a DIRV method without the error correction.
DIRV is a method that has both stabilization techniques.
The hyper-parameter $\gamma$ was set to 1.0.

TDM could not be used in the real service setting because there are cases where the rankings generated by the TDM might not exist in the real service data, as discussed in \ref{realworlddata}.
Therefore, TDM was only used for the simulation-based experiments.
In the DIRV methods for the real service setting, we generated rankings that minimize variance using $f(o)$ and $g(o)$ from the candidates of the interleaved rankings.

\section{Results and Discussion}
\label{sec:results}

\begin{figure*}[t]
\begin{center}
\begin{tabular}{ccc}

\begin{minipage}{0.3\hsize}
\begin{center}
\includegraphics[scale=0.3]{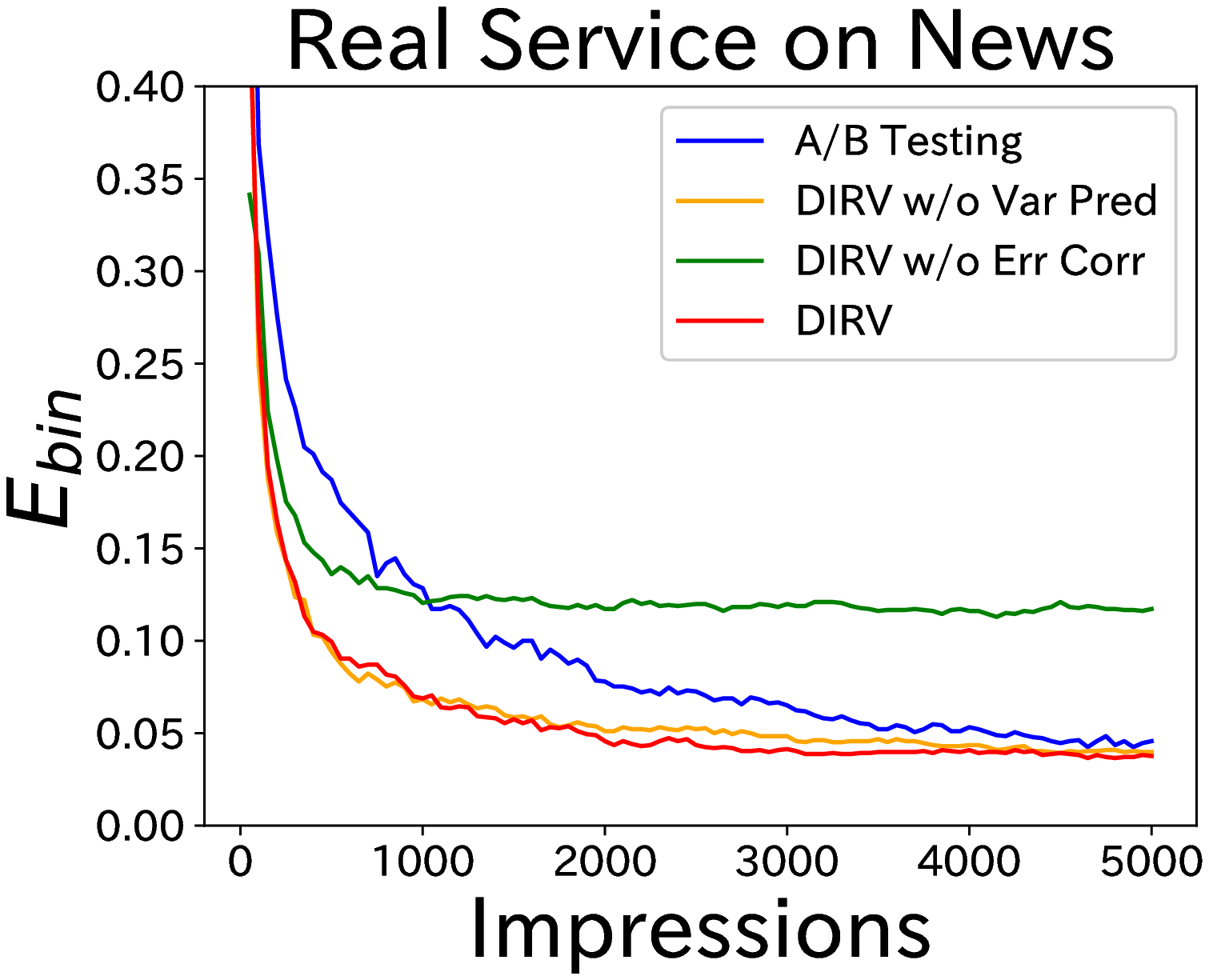}
\caption{$E_{\rm bin}$ averaged over the number of impressions of the real service News dataset. DIRV or DIRV w/o Var Pred had the lowest $E_{\rm bin}$ for each impression.}
 \label{fig:realefficiencty}
\end{center}
\end{minipage}
&
\begin{minipage}{0.3\hsize}
\begin{center}
\includegraphics[scale=0.3]{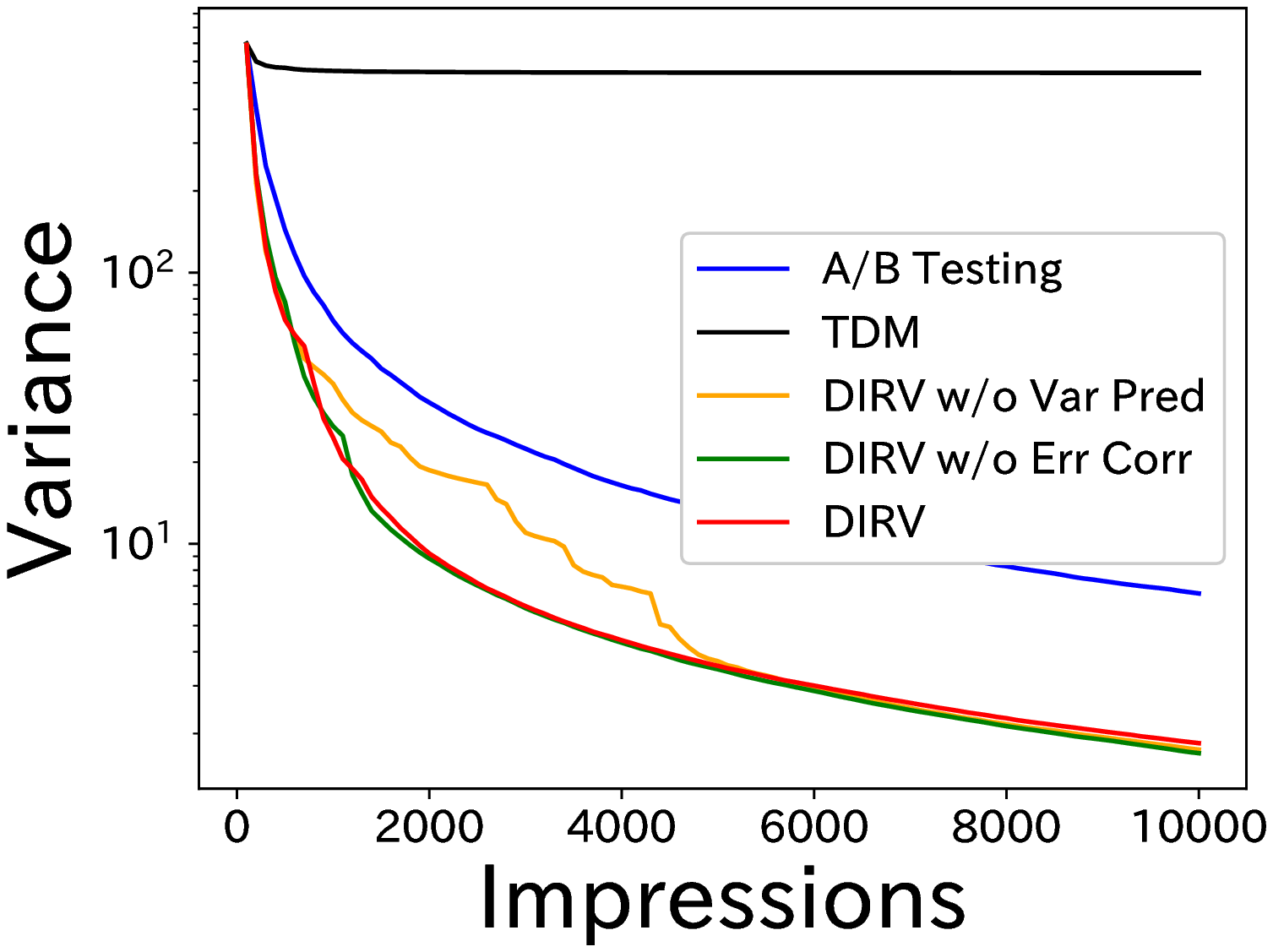}
\caption{Variance versus the number of impressions in the simulation-based News dataset. The result shows that DIRV had lower variance for each impression, especially for a small number of impressions.}
\label{fig:variance} 
\end{center}
\end{minipage}
&
\begin{minipage}{0.3\hsize}
\begin{center}
\includegraphics[scale=0.25]{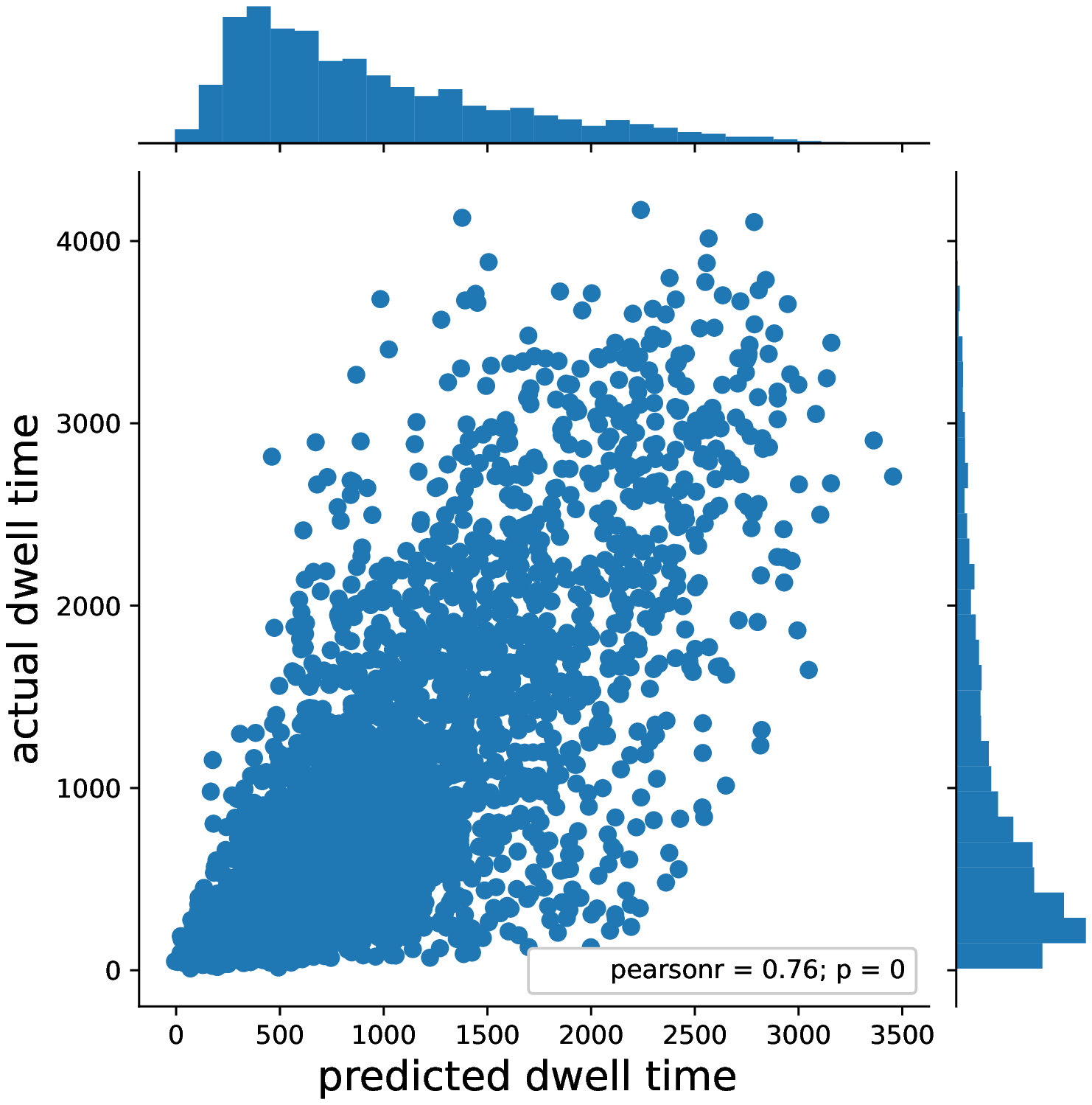}
\caption{Visualization of the predicted variance. Each dot corresponds to the predicted and actual variance values for each News article. The Pearson's correlation value was $0.76$.}
\label{fig:varianceprediction} \end{center}
\end{minipage} 

\end{tabular}
\end{center}
\end{figure*}

In this section, we discuss the experimental results for answering the research questions posed in section \ref{researchquestion}.

\subsection{RQ1: Can DIRV identify preferences
between rankings more efficiently and accurately than comparison methods?}

\subsubsection{Efficiency}
\figref{efficiencty} shows the efficiency results for the simulation-based setting for the News, LETOR, and EC datasets.
\figref{realefficiencty} shows the efficiency for the real service setting.
In Figures \ref{fig:efficiencty} and \ref{fig:realefficiencty}, the x-axis represents the number of impressions, and the y-axis represents $E_{\rm bin}$.
For all of the datasets, DIRV had the lowest $E_{\rm bin}$ for each impression.
In contrast to the simulation-based setting, DIRV w/o Err Corr had higher $E_{\rm bin}$ after $1,000$ impressions compared to the results from the A/B testing in the real service setting.
TDM stopped decreasing $E_{\rm bin}$ during the early stage in all datasets.

Among the evaluated methods, DIRV achieved the highest efficiency by reducing the variances.
DIRV was designed to reduce variance by (1) aggregating post-clicks from different rankings that leads to an increased sample size for each item and (2) exposing the item that has high variance.
DIRV successfully reduced the variance, which led to high efficiency (as shown in \figref{variance}).

\subsubsection{Accuracy}
\tabref{resaccuracy} details the evaluation accuracy results.
The number closest to 0.0 for each parameter and dataset are highlighted in bold.
In the LETOR dataset, there were eight $E_{\rm bin}$ results for each dataset.
In the News and EC datasets, there were five $E_{\rm bin}$ results based on the item duplication ratios (which ranged from 0\% to 80\%, in 20\% increments).
\figref{realefficiencty} shows the accuracy of the real service setting in the last impression (i.e., at 5,000 impressions).

Overall, DIRV outperformed the existing methods for all of the datasets.
A/B testing resulted in the lowest $E_{\rm bin}$ in the EC dataset with duplication ratios of 0\%.
We designed the estimator using the expectation for post-click metrics.
DIRV achieved high accuracy compared with TDM due to the use of our estimator. 
It is remarkable that our method performed well when many items were shared among the input rankings (i.e., high  duplication ratio).
Similar to the results from \cite{schuth2015predicting}, our results demonstrated that TDM is difficult to extend for aggregating continuous values like post-click values.

Regarding {\bf RQ1}, DIRV outperformed the existing methods in efficiency and accuracy. 
The performance was especially remarkable when many items were shared among the input rankings.

\subsection{RQ2: How does the variance prediction technique affect the evaluation efficiency?}
\figref{variance} details the variance reduction on the News dataset in the simulation setting.
The results show that DIRV achieved the lowest variance among the evaluated methods.
DIRV w/o Var Pred decreased the variance slowly compared to DIRV, especially for a small number of impressions.
TDM had the highest variance for each impression.
These trends were also observed in the other datasets.

\figref{variance} shows that the predicted variance contributed to reducing the variance in small impressions.
In contrast, the variance of some items was underestimated in DIRV w/o Var Pred and could not obtain a sufficient number of exposures during the early stage.
This resulted in reducing the variance in the post-click metrics slowly.
The variance prediction technique reduced the variance and improved efficiency.
In TDM, some items at the bottom of the input ranking was not selected for the interleaved ranking, which led to high variance.

\figref{realefficiencty} shows that the variance prediction did not contribute to the efficiency in the real service setting.
This is because the total number of items was small in this setting.
Thus, there were fewer benefits from manipulating the order of the items using the predicted variance.

Regarding {\bf RQ2}, the variance prediction techniques contributed to reducing the variance, which led to improved efficiency.

\subsection{RQ3: How does the error correction technique affect the evaluation accuracy?}
\figref{realefficiencty} shows the accuracy of the real service setting at the end of the evaluation (i.e., at 5,000 impressions).
The results show that the accuracy was close to 0.0 for DIRV and A/B testing methods.
In contrast, the DIRV w/o Err Corr method stopped decreasing the $E_{\rm bin}$ after 1,000 impressions.

The results show that the error was reduced as the number of impressions increased in the real service setting by introducing an error correction technique.
In particular, we successfully confirmed that the evaluation error converged around at 0.0 using the error correction technique, which reduced the systematic error as the experiment progressed.
This result supported the assumption that the post-click behavior was independent of other user behavior in this real service setting. 

The difference between the assumed cascade click model and the actual user behavior in the real service setting resulted in a systematic error.
\figref{realefficiencty} illustrated the error in the results of the DIRV w/o Err Corr method.
In the simulation-based setting, we assumed that the actual user behavior also obeyed the cascade click model.
This assumption led to the same accuracy between the DIRV and DIRV w/o Err Corr methods in \tabref{resaccuracy}.

Regarding {\bf RQ3}, we reduced the systematic error using the error correction technique.
The error correction technique led to improved accuracy when an estimation error existed in the click model.

\section{Conclusion and Future Work}
\label{sec:conclusion}
In this study, we proposed a method for accurate and efficient post-click evaluation using interleaving.
First, we introduced the click model to aggregate the post-click values.
Then, we showed that minimizing the variance of post-click metrics leads to a reduction in the evaluation error.
Next, we provided a policy to generate rankings to reduce the variance.
Finally, we proposed two stabilization techniques to support the proposed method.

To evaluate this method, we conducted comprehensive experiments with both simulation-based and real service settings.
The experimental results from this study indicated that 1) the proposed method outperformed existing methods in efficiency and accuracy and the performance was especially remarkable when many items were shared among the input rankings,
2) the variance prediction techniques contributed to reducing the variance, which led to improved efficiency,
and 3) we could successfully reduce the systematic error using the error correction technique, which led to improved accuracy when there existed an estimation error in the click model.

In this study, we focused on building an online evaluating method.
In the future, we plan to extend our framework to include a bandit algorithm in ranking optimization.

\section*{Appendix}
\appendix

\section{Variance of $\overline{E}[x|r_i]$}
\label{dep}

Noting that $V[x + y] = V[x] + V[y]$ if $x$ and $y$ are independent,
we obtain:
\[
V[\overline{E}[x|r_i]] = V \left[ \sum_{d \in r_i} \overline{P}(c_d=1|r_i) \overline{E}[x|d] \right] =  \sum_{d \in r_i} V \left[ \overline{P}(c_d=1|r_i) \overline{E}[x|d] \right] \nonumber
\]

Using $V[xy] = V[x]V[y] + E[x]^2V[y] + V[x]E[y]^2$ ($x$ and $y$ are independent),
the variance for each item is obtained as follows:
\begin{eqnarray}
V \left[ \overline{P}(c_d=1|r_i) \overline{E}[x|d] \right] 
&=& V[\overline{P}(c_d=1|r_i)] V[\overline{E}[x|d]] \nonumber \\
&+& E[\overline{P}(c_d=1|r_i)]^2 V[\overline{E}[x|d]] \nonumber \\
&+& E[\overline{E}[x|d]]^2 V[\overline{P}(c_d=1|r_i)] \nonumber
\end{eqnarray}
Since $V[\overline{x}] = \sigma^2/n$ ($\overline{x}$ is the sample mean, $\sigma^2$ is the population variance,
and $n$ is the sample size),
$V[\overline{P}(c_d=1|r_i)]$ can be computed as 
$
V[\overline{P}(c_d=1|r_i)] = V[P(c_d=1|d)]/n^i_d,
$
where $n^{i}_d$ is the number of impressions of item $d$.
We can also obtain $V[\overline{E}[x|d]] = V[x|d] / n^c_d$ where $n^c_d$ is the number of clicks on item $d$.

Finally, we express the variance of $\overline{E}[x|r_i]$
as a function of the samples sizes $n^i_d$ and $n^c_d$:
\begin{equation*}
\begin{split}
V[\overline{E}[x|r_i]]
&=\sum_{d \in r_i} \left\{ \frac{V[P(c_d=1|d)]}{n^i_d}
\frac{V[x|d]}{n^c_{d}} \right. 
+ E[\overline{P}(c_d=1|r_i)]^2 \frac{V[x|d]}{n^c_{d}} \\
&+ \left. E[\overline{E}[x|d]]^2  \frac{V[P(c_d=1|d)]}{n^i_d} \right\} 
 \coloneqq \sum_{d \in r_i} \phi_{d, r_i}(n^i_d, n^c_d) \nonumber
\end{split}
\end{equation*}

The function $\phi_{d, r_i}$ monotonically decreases for $n^i_d$ and $n^c_d$ when $x$ is always positive.
As $\overline{P}(c_d=1|r_i)$ and $\overline{E}[x|d]$
are unbiased estimators,
$E[\overline{P}(c_d=1|r_i)]$ and $E[\overline{E}[x|d]]$
can be approximated by $\overline{P}(c_d=1|r_i)$
and $\overline{E}[x|d]$ with a sufficient number of samples.
Assuming that $P(a_d|d)$ follows a Bernoulli distribution,
we can also approximate $V[P(c_d=1|d)]$ by $\overline{P}(c_d=1|r_i) (1-\overline{P}(c_d=1|r_i))$.

\bibliographystyle{ACM-Reference-Format}
\balance
\bibliography{sample-base}

\end{document}